\documentclass[review,10pt]{elsarticle}

\usepackage{amssymb}
\usepackage{amsthm}
\usepackage{amsmath}
\usepackage{yhmath}
\usepackage{bm}
\usepackage{mathtools}
\usepackage{varwidth}

\usepackage{graphicx}
\usepackage{float}
\usepackage{graphicx}
\usepackage{float}
\usepackage{subfig}
\allowdisplaybreaks[4]
\biboptions{sort&compress}

\journal{}

\begin{document}

\begin{frontmatter}


\title{Prescribed Time Time-varying Output Formation Tracking for Uncertain Heterogeneous Multi-agent Systems}


\author[author1]{Binghe An}
\author[author1]{Bo Wang\corref{cor1}}
\author[author1]{Huijin Fan\corref{cor1}}
\cortext[cor1]{Corresponding author}


\author[author1]{Lei Liu}
\author[author1]{Yongji Wang}
\address{National Key Laboratory of Science and Technology on Multispectral Information Processing, School of Artificial Intelligence and Automation,
Huazhong University of Science and Technology, Wuhan 430074, PR China}

\newtheorem{theorem}{Theorem}
\newtheorem{assumption}{Assumption}
\newtheorem{lemma}{Lemma}
\newtheorem{remark}{Remark}
\newtheorem{definition}{Definition}
\allowdisplaybreaks[4]
\begin{abstract}
The time-varying output formation tracking for the heterogeneous multi-agent systems (MAS) is investigated in this paper. First, a distributed observer is constructed for followers to estimate the states of the leader, which can ensure that the estimation error converges to the origin in the prescribed time. Then, the local formation controller is designed for each follower based on the estimation of the observer, under which, the formation errors converge to the origin in the prescribed time as well. That is, the settling time of the whole system can be predefined in advance. It's noted that not only the uncertainties in the state matrix but also the uncertainties in the input matrix are considered, which makes the problem more practical. Last, a simulation is performed to show the effectiveness of the proposed approach.
\end{abstract}



\begin{keyword}
Fixed-wing UAVs \sep formation tracking \sep velocity constraint \sep  fixed time distributed observer \sep R-CBF.


\end{keyword}
\end{frontmatter}



\section{Introduction}
Due to its wide applications in different fields, such as unmanned aerial vehicles \cite{Bhowmick2022}, surface vehicles \cite{Ringback2022}, and spacecraft \cite{Wei2022}, formation control of MAS has been a hot research topic in recent years. It's noted that in some scenes, the networked agents may have different dynamics. Therefore, some significant research has tried to address the cooperative formation control of the heterogeneous MAS with the state feedback \cite{Yan2021, Shi2022} and the output feedback \cite{HuaY2019, HuaY2020, Song2022}.\par
It's worth mentioning that most existing results on formation control of heterogeneous MAS pay attention to achieving the desired formation configuration when the time tends to infinity. However, in the practical applications that have to be completed in finite time, such as the multi-missile guidance, the asymptotical convergence can not meet the requirement of a fast convergence rate. Hence, the finite time cooperative output regulation has been investigated to speed up the convergence rate \cite{Wu2022}, which can guarantee that the tracking error converges to the origin in a finite time related to the initial states. Further, to remove the dependence of the settling time on the initial states, the fixed time control strategies have been proposed to stabilize the error system, which allows that the settling time is independent of the initial states \cite{Cai2021}. Although the upper bound of the settling time of the fixed time control approaches can be explicitly given, it depends on multiple parameters in the controller thus inconvenient to be predefined in advance. Recently, the prescribed time control with a time-varying gain has received increasing attention own to its benefit on the convergence time \cite{Wang2019}, which can be prespecified as desired by only one parameter in the controller. \par
On the other hand, there are usually unmodeled dynamics and uncertain parameters in the practical system. The uncertainties may damage the performance of the controlled system and even destabilize the system. Although the formation control for heterogeneous MAS was discussed in \cite{HuaY2019, Shi2022, Song2022, Cai2021}, no uncertainties in the systems were considered. The formation control of heterogeneous MAS with uncertainties was the state matrix is studied in \cite{Shi2022}, while the uncertainties in the input matrix are ignored. To be more practical, it's meaningful and significant to develop new methods that can realize the formation tracking for the heterogeneous MAS with uncertainties in the both state and input matrices while ensuring that the formation errors converge in a prescribed time.\par
Based on the above discussions, the time-varying output formation tracking of uncertain heterogeneous MAS is investigated in this paper. A prescribed time control strategy composed of a distributed observer and a local formation controller is constructed, which guarantees that the output of the followers can track that of the leader and the formation errors converge to the origin in the prescribed time despite the external disturbances and uncertainties. The contributions of this paper can be summarized as follows:\\
(1) A distributed observer is constructed to estimate the states of the leader, by which, the estimation errors converge to the origin in finite time despite the leader's unknown input and the convergence time can be specified in advance by only one parameter in the observer.\\
(2) A prescribed time formation controller is proposed for each follower with the help of a time scaling function such the formation errors can disappear after a given time. Moreover, robust terms are designed in the controller to suppress the effect of external disturbances and the unknown input of the leader.\\
(3) The uncertainties in the both state matrix and input matrix are considered, which makes the problem more practical while challenging. Compensation items to deal with uncertainties are constructed in the controller to ensure the stability of the controlled system.\par
The content of this paper is organized as follows. Preliminaries and the problem formulation are given in Section 2. Then, the prescribed time formation control strategy is described in detail in Section 3. Simulation is given in Section 4 to demonstrate the effectiveness of the proposed control approach. Section 5 is reserved for the conclusion.\\
\textbf{Notations}. In what follows, ${R^n}$ means the set of $n$-dimensional real vectors. $\left\|  \cdot  \right\|$ represents the 2-norm of a matrix. ${\lambda _{\min }}( \cdot )$  and ${\lambda _{\max }}( \cdot )$  denote the minimum and maximum eigenvalues of a matrix, respectively. $diag\{ {a_i}\} $  represents a diagonal matrix with $a_i$  as its diagonal element. For $x \in R^n$, $\text{sgn}(x)=\frac{x}{\left\|  x \right\|}$.
\section{Preliminaries and Problem Formulation}
\subsection{Graph Theory}
The considered multi-agent systems consist of one leader and $N$ followers. The leader acquires no information from followers and is marked as agent 0. The interaction network among followers is described by a graph $G(S,E)$, where $S$ represents the set of followers and $E$ denotes the set of edges in the graph  $G(S,E)$. The $A = [{a_{ij}}] \in {R^{N \times N}}$  and $L = [{l_{ij}}] \in {R^{N \times N}}$  are the adjacency matrix and Laplacian matrix of  $G$, respectively. The detailed definitions of  $A$ and  $L$ can be seen in \cite{An2022}, thus omitted here. $B = diag\{ {b_i}\}  \in {R^{N \times N}}$ is a diagonal matrix with  $b_i\;(i=1,...,N)$ as its diagonal elements. $b_i=1$ if follower $i$ has access to the {leader's} information and $b_i=0$ if not. The matrix $H=L+B$ is utilized to describe the communication network of the whole leader-follower multi-agent system.

\begin{assumption}
The communication among followers is undirected and there exists at least one path from the leader to each follower.
\end{assumption}
\subsection{Problem Formulation}
The dynamics of follower $i\;(i=1,...,N)$ is described by
\begin{align}\label{model_follower}
{{\dot x}_i}& = ({A_i} + \Delta {A_i}){x_i} + ({B_i} + \Delta {B_i})({u_i} + {d_i}),\notag \\
{y_i} &= {C_i}{x_i},
\end{align}
where ${x_i} \in {R^{{n_i}}}$ and ${y_i} \in {R^{{q}}}$ are the state and the output of follower $i$,  respectively. ${A_i} \in R^{n_i \times n_i}$ and ${B_i}\ \in R^{n_i \times m_i}$ are constant matrics.  ${d_i}$ represents the external disturbances and ${u_i}$  is the control input. $\Delta {A_i} \in R^{n_i \times n_i}$ and $\Delta {B_i}\ \in R^{n_i \times m_i}$ denote the unknown, time-varying uncertainties in the state matrix and input matrix, respectively. It is assumed that $\Delta {A_i}$ and $\Delta {B_i}$ satisfy the following condition:
\begin{align}
\Delta {A_i} = {B_i}{N_i}(t), \, \Delta {B_i} = {B_i}{M_i}(t),
\end{align}
where ${N_i}(t) \in R^{m_i \times n_i}$ is an unknown, time-varying while norm-bounded matrix, that is, there exists a known constant $\bar N$ satisfying $\left\| {{N_i}(t)} \right\| \le \bar N$. ${M_i(t)} =diag\{{w_1(t), ...,w_{m_i}(t)}\} \in R^{m_i \times m_i}$ is an unknown diagonal matrix and the absolute values of whose diagonal elements are less than $\bar M$  with $ 0\leq \bar M  <1$, that is $|w_i(t)|<\bar M<1\,(i=1,...,m_i)$.\par
The dynamics of the leader is described by
\begin{align}\label{leader_model}
{{\dot x}_0} & = {A_0}{x_0} + {B_0}{u_0},\notag \\
{y_0} &  = {C_0}{x_0},
\end{align}
where ${x_0} \in {R^{{n_0}}}$ and ${y_0} \in {R^{{q}}}$ are the state and the output of the leader. ${A_0} \in R^{n_0 \times n_0}$ and ${B_0}\ \in R^{n_0 \times m_0}$. ${u_0}$ is the control input of the leader, which is unknown to all followers.
\begin{assumption}
The disturbances $d_i\;(i=1,...,N)$ and the input of the leader $u_0$ are bounded, i.e., there are positive constants  $\bar d$ and $\bar u$ such that $\left\| {{u_0}} \right\| \le \bar u$ and $\left\| {{d_i}} \right\| \le \bar d$  hold.
\end{assumption}

\begin{assumption}
The matrix  $B_i \, (i=1 ,..., N)$ is row-full rank.
\end{assumption}

\begin{assumption}
For all $ i=1,..., N$, the following regulation equation has solution pairs $({X_i},{\rm{ }}{U_i})$.
\begin{align}
{X_i}{A_0} &= {A_i}{X_i} + {B_i}{U_i},\notag \\
{C_i}{X_i} & = {C_0}.
\end{align}
\end{assumption}
\subsection{Related Lemmas}

\begin{lemma}\cite{Wang2019, Xu2021} \label{prescribed_time}
Considering the following system
\begin{align}\label{general_system}
\dot x(t) = g(t,x(t)),\;\;\;{\kern 1pt}\;{\kern 1pt} x(0) = {x_0}
\end{align}
where $x(t)$ is the state vector and $g(t,x(t))$ is a nonlinear vector field locally bounded uniformly in time. Suppose there is a Lyapunov function $V(t)$ such that
\begin{align}
\frac{{dV(x(t))}}{{dt}} \le  - aV - b\frac{{\dot \mu (t)}}{{\mu (t)}}V(x(t)),
\end{align}
where $a$ and $b$ are positive constants and $\mu (t)$ is a time scaling function defined as
\begin{align}
{\mu }(t) = \left\{ {\begin{array}{*{20}{l}}
{\frac{{T^v}}{{{{\left( {{T} + {t_0} - t} \right)}^v}}},}&{t \in \left[ {{t_0},{T} + {t_0}} \right)}\\
{1,}&{t \in \left[ {{T} + {t_0},\infty } \right)}
\end{array}} \right.
\end{align}
where $v>2$ is a constant, ${t_0}$ is the initial time and $T>t_0$. Then, the origin of the system (\ref{general_system}) is globally prescribed-time stable with the prescribed time $T$.
\end{lemma}
\begin{lemma}
 \cite{AnB2022} Under Assumption 2, $H$ is a symmetric and positive definite matrix.
\end{lemma}\par

The desired time-varying output formation pattern of follower $i$ is described by a time-varying vector $h_i(t)$. The control objective is to design a distributed formation controller $u_i$ such that the following requirements can be satisfied
\begin{align*}
\mathop {\lim }\limits_{t \to \bar T} \left\| {{y_i} - {y_0} - {h_i}} \right\| = 0, \; \\
 \left\| {{y_i} - {y_0} - {h_i}} \right\| = 0, \;\text{for} \;  t>\bar T
\end{align*}
where $\bar T$ is the prescribed time, which can be specified by one parameter in the controller.\par

\section{Main Result}
In this section, the distributed formation strategy is derived in detail. First, a prescribed time observer using neighbors' information is constructed for followers, which can provide an accurate estimation of leader's states. Then, the formation controller is designed to realize the desired formation pattern in the prescribed time despite the uncertainties and disturbances.
\subsection{The Design of the Prescribed Time Observer}
For the follower $i\;(i=1,...,N)$, the prescribed time observer is constructed as follows
\begin{align}\label{observer}
\dot{\xi _i} = &{A_0}{{ \xi }_i} - \frac{{{c_i}}}{2}{B_0}B_0^T{P_0}\Big(\sum\limits_{j = 1}^N {{a_{ij}}} {\rm{(}}{\xi _i} - {\xi _j}{\rm{)}} + {b_i}({\xi _i} - {x_0})\Big) \notag \\
&- {\beta _i}\frac{{{{\dot \mu_1(t) }}}}{{{\mu _1(t)}}}\Big(\sum\limits_{j = 1}^N {{a_{ij}}} {\rm{(}}{\xi _i} - {\xi _j}{\rm{)}} + {b_i}({\xi _i} - {x_0})\Big)\\ \notag
&- {\alpha _i}{\mathop{\rm sgn}} \Big({P_0}\Big(\sum\limits_{j = 1}^N {{a_{ij}}} {\rm{(}}{\xi _i} - {\xi _j}{\rm{)}} + {b_i}({\xi _i} - {x_0})\Big)\Big),
\end{align}
where ${c_i}$ , $\alpha_i$ and $\beta_i$ are positive constants. $P_0$ is the solution of the following Riccati equation
\begin{align}\label{Riccati_P0}
A_0^T{P_0} + {P_0}{A_0} - {P_0}{B_0}{B_0}^T{P_0} + {I_{{n_0}}} = 0.
\end{align}
$\mu_1(t)$ is a time-scaling function defined as
\begin{align}
{\mu _1}(t) = \left\{ {\begin{array}{*{20}{l}}
{\frac{{T_1^v}}{{{{\left( {{T_1} + {t_0} - t} \right)}^v}}},}&{t \in \left[ {{t_0},{T_1} + {t_0}} \right)}\\
{1,}&{t \in \left[ {{T_1} + {t_0},\infty } \right)}
\end{array}} \right.
\end{align}
where $T_1$ is the prescribed convergence time for the observer (\ref{observer}).\par
Define the estimation error $\tilde \xi _i$ as ${\tilde \xi _i} = {\xi _i} - {x_0}$, then, the following theorem is given to show the convergence of the observer (\ref{observer}).
\newtheorem{thm}{\bf Theorem}
\begin{theorem}\label{thm1}
Suppose the Assumptions 1-3 hold and the parameters are chosen to satisfy
\begin{align}\label{parameter_controller}
{c_i} \ge \frac{1}{{{\lambda _{\min }}(H)}},\;{\alpha _i} \ge \left\| B \right\|\bar u,\;\,\quad \text{for}\; \, i=1,...N
\end{align}
then the distributed observer (\ref{observer}) can ensure that the estimation errors $\tilde \xi _i (i=1,...,N) $ converge to the origin in the prescribed time $T_1$.
\end{theorem}
\begin{proof}
Define ${\eta _i} = \sum\limits_{j = 1}^N {{a_{ij}}} \left( {{\xi _i} - {\xi _j}} \right) + {b_i}\left( {{\xi _i} - {x_0}} \right)$, then, (\ref{observer}) can be rewritten in the following compact form
\begin{align}
\dot \xi  = & {\rm{ }}\left( {{I_N} \otimes {A_0}} \right)\varepsilon  - \frac{1}{2}CH \otimes {B_0}B_0^T{P_0}\tilde \xi  \notag \\
& - \frac{{{{\dot \mu_1(t) }}}}{{{\mu _1(t)}}}\left( {\beta H \otimes {I_n}} \right)\tilde \xi  - (\alpha  \otimes {I_N})f(\eta ),
\end{align}
where $\xi  = {[\xi _1^T,...,\xi _N^T]^T} \in {R^{N{n_0}}}$, $\tilde \xi  = {[\tilde \xi _1^T,...,\tilde \xi _N^T]^T} \in {R^{N{n_0}}}$, $C = diag\{ {c_1},...,{c_N}\}  \in {R^{N \times N}}$, $\beta  = diag\{ {\beta _1},...,{\beta _N}\}  \in {R^{N \times N}}$ and $\alpha  = diag\{ {\alpha _1},...,{\alpha _N}\}  \in {R^{N \times N}}$, $f(\eta ) = [\text{sgn}^T(P_0\eta_1),...,\text{sgn}^T(P_0\eta_N)]^T$.\par
According to (\ref{leader_model}) and (\ref{observer}), it yields that
\begin{align}\label{estimation_error}
\dot {\tilde \xi}  = &{\rm{ }}\left( {{I_N} \otimes {A_0}} \right)\tilde \xi  - \frac{1}{2}CH \otimes {B_0}B_0^T{P_0}\tilde \xi  \notag \\
 &  -  \frac{{{{\dot \mu }_1(t)}}}{{{\mu _1}(t)}}\left( {\beta H \otimes {I_n}} \right)\tilde \xi- (\alpha  \otimes {I_N})f(\eta ) - ({I_N} \otimes B{u_0}).
\end{align}\par
Now, the following Lyapunov candidate is chosen:
\begin{align}\label{V}
V = {\tilde \xi ^T}\left( {H \otimes {P_0}} \right)\tilde \xi.
\end{align}\par
Differentiating (\ref{V}) along the trajectory (\ref{estimation_error}) gives
\begin{align}\label{dv}
\dot V = & 2{{\tilde \xi }^T}\left( {H \otimes {P_0}} \right)[(I_N \otimes {A_0})\tilde \xi  - \frac{1}{2}CH \otimes {B_0}B_0^ T {P_0}\tilde \xi  \notag \\
&- \frac{{{{\dot \mu }_1}(t)}}{{{\mu _1}(t)}}\left( {\beta H \otimes {I_n}} \right)\tilde \xi - (\alpha  \otimes {I_n})f(\eta ) - {I_N} \otimes B{u_0}]\notag \\
{\rm{     = }}&{{\tilde \xi }^ T }(H \otimes (A_0^T{P_0} + {P_0}{A_0}) - HC{H} \otimes {P_0}{B_0}B_0^ T {P_0})\tilde \xi  \notag \\
& - 2\frac{{{{\dot \mu }_1}(t)}}{{{\mu _1}(t)}}{{\tilde \xi }^T}\left( {H\beta H \otimes {P_0}} \right)\tilde \xi \notag \\ &- 2{{\tilde \xi }^T}\left( {H \otimes {P_0}} \right)({I_N} \otimes B{u_0})\notag \\
& - 2{{\tilde \xi }^T}\left( {H \otimes {P_0}} \right)(\alpha  \otimes {I_n})f(\eta ).
\end{align}\par
Per Lemma 2, there exists an orthogonal matrix $Q$ such that
\begin{align*}
{Q^{ T}}HQ = \Lambda  = diag\{ {\lambda _1},...,{\lambda _N}\},
\end{align*}
where ${\lambda _i} (i=1,...,N)$ are the eigenvalues of the matrix $H$. \par
Thus, one has that
\begin{align}\label{duijiaohua}
&{{\tilde \xi }^T}(H \otimes (A_0^T{P_0} + {P_0}{A_0}) - HCH \otimes {P_0}{B_0}B_0^T{P_0})\tilde \xi \notag \\
\leq & {{\tilde \xi }^T}(H \otimes (A_0^T{P_0} + {P_0}{A_0}) - {\lambda _{\min }}(C){H^2} \otimes {P_0}{B_0}B_0^T{P_0})\tilde \xi \notag \\
\leq &{{\tilde \xi }^T}(H \otimes (A_0^T{P_0} + {P_0}{A_0} \notag\\
& -{\lambda _{\min }}(C){\lambda _{\min }}(H) \otimes {P_0}{B_0}B_0^T{P_0}))\tilde \xi \notag \\
 \leq & {{\bar \xi}^T}(\Lambda  \otimes (A_0^T{P_0} + {P_0}{A_0} \notag \\
 &- {\lambda _{\min }}(C){\lambda _{\min }}(H) \otimes {P_0}{B_0}B_0^ T {P_0}))\bar \xi,
\end{align}
where $\bar \xi  = ({Q} \otimes {I_n})\tilde \xi$.\par
Due to (\ref{parameter_controller}) and (\ref{duijiaohua}), it gives
\begin{align}\label{duijiaohua2}
&\quad {{\tilde \xi }^T}(H \otimes (A_0^T{P_0} + {P_0}{A_0}) - HCH \otimes {P_0}{B_0}B_0^T{P_0})\tilde \xi \notag\\
 &\le {{\bar \xi }^T}(\Lambda  \otimes  - I)\bar \xi  \le  - {\lambda _{\min }}(H){{\bar \xi }^T}\bar \xi \notag\\
 & \le  - \frac{{{\lambda _{\min }}(H)V}}{{{\lambda _{\max }}(H \otimes {P_0})}}.
\end{align}\par
Substituting (\ref{duijiaohua2}) into (\ref{dv}) yields
\begin{align}\label{fangsuo33}
\dot V \le & - \frac{{{\lambda _{\min }}(H)V}}{{{\lambda _{\max }}(H \otimes {P_0})}} - 2\frac{{{{\dot \mu }_1}(t)}}{{{\mu _1}(t)}}{{\tilde \xi }^T}\left( {H\beta H \otimes {P_0}} \right)\tilde \xi \notag \\
& - 2{{\tilde \xi }^T}\left( {H \otimes {P_0}} \right)(\alpha  \otimes {I_n})f(\eta )\notag \\
& - 2{{\tilde \xi }^T}\left( {H \otimes {P_0}} \right)({I_N} \otimes B{u_0}).
\end{align}\par
Further, one has that
\begin{align}\label{fangsuo11}
&\quad \;2\frac{{{{\dot \mu }_1}(t)}}{{{\mu _1}(t)}}{{\tilde \xi }^T}\left( {H\beta H \otimes {P_0}} \right)\tilde \xi \notag \\
& \ge 2{\lambda _{\min }}(\beta ){\lambda _{\min }}(H)\frac{{{{\dot \mu }_1}(t)}}{{{\mu _1}(t)}}{{\tilde \xi }^T}\left( {H \otimes {P_0}} \right)\tilde \xi \notag \\
 &= 2{\lambda _{\min }}(\beta ){\lambda _{\min }}(H)\frac{{{{\dot \mu }_1}(t)}}{{{\mu _1}(t)}}V,
 \end{align}
 and
 \begin{align}\label{fangsuo22}
 & - 2{{\tilde \xi }^T}(H \otimes {P_0})(\alpha  \otimes {I_n})f(\eta ) - 2{{\tilde \xi }^T}(H \otimes {P_0})({I_n} \otimes B{u_0})\notag \\
 & = 2[ - {f^T}(\eta )(\alpha H \otimes {P_0})\tilde \xi  - {({I_n} \otimes B{u_0})^T}(H \otimes {P_0})\tilde \xi ]\notag \\
 & \le 2\sum\limits_{i = 1}^N \Big[{\frac{{ - {\alpha _i}{{({P_0}{\eta _i})}^T}{P_0}{\eta _i}}}{{\left\| {{P_0}{\eta _i}} \right\|}}}  + \left\| {B{u_0}} \right\|\left\| {{P_0}{\eta _i}} \right\| \Big]\notag \\
 &= 2\sum\limits_{i = 1}^N {\left[ { - ({\alpha _i} - \left\| {B{u_0}} \right\|)\left\| {{P_0}{\eta _i}} \right\|} \right]} \notag \\
 & \le 2\sum\limits_{i = 1}^N {\left[ { - ({\alpha _i} - \left\| {{B_0}} \right\|\bar u)\left\| {{P_0}{\eta _i}} \right\|} \right]}  \le 0.
 \end{align}\par
Substituting (\ref{fangsuo11}) and (\ref{fangsuo22}) into (\ref{fangsuo33}), one has that
\begin{align}
\dot V \le  - \frac{{{\lambda _{\min }}(H)}}{{{\lambda _{\max }}(H \otimes {P_0})}}V - 2{\lambda _{\min }}(\beta ){\lambda _{\min }}(H)\frac{{{{\dot \mu }_1}(t)}}{{{\mu _1}(t)}}V.
\end{align}\par
Per Lemma 1, it is obtained that the estimation error $\tilde \xi $ converges to the origin in the prescribed time $T_1$.\par
The proof is therefore completed.
\end{proof}

\subsection{The Design of the Prescribed Time Formation Controller}
In this subsection, the local formation controller is proposed for each follower to achieve the desired time-varying formation.\par
Under Assumption 3, there is a matrix $B_i^+ \in R^{m_i\times n_i}$, such that  ${B_i}B_i^ +  = {I_{_{{n_i}}}}$. Furhter, there exists a matrix $F_i= B_i^ + {X_i}{B_0}$ as the solution of the equation ${X_i}{B_0} = {B_i}{F_i}$.\par

For follower $i$, the desired time-varying formation pattern   is generated by the following local exosystem
\begin{align}\label{formation_system}
{{\dot {\tilde h}}_{i}} = {A_{hi}}{{\tilde h}_{i}},\notag \\
{h_{i}} = {C_{hi}}{{\tilde h}_{i}},
\end{align}
where $\tilde h_i \in R^{n_{hi}}$, $A_{hi} \in R^{n_{hi}\times n_{hi}}$ and $C_{hi} \in R^{q \times n_{hi}}$.
\begin{assumption}
For all $i=1,..., N$, the following regulation equation has solution pairs $({X_{hi}},{\rm{ }}{U_{hi}})$.
\begin{align*}
X_{hi}A_{hi} &= A_{i}X_{hi} + B_{i}U_{hi},\notag \\
{C_i}{X_{hi}} & = {C_{{hi}}}.
\end{align*}
\end{assumption}
To achieve the desired formation configuration, the formation controller $u_i\;(i=1,...,N)$ is constructed as follows
\begin{align}\label{u}
{u _i}(t) = \left\{ {\begin{array}{*{20}{l}}
{\textbf{0},}&{t \leq T_1}\\
{u_{i,1}+u_{i,2}+u_{i,3},}&{t >T_1}
\end{array}} \right.
\end{align}
where the expression of $u_{i,1}$ , $u_{i,2}$ and $u_{i,3}$ are given as
\begin{align}\label{ui1}
{u_{i,1}} = {K_{i,1}}{x_i} + {K_{i,2}}{\xi _i} + {K_{i,3}}{\tilde h_i} - {K_{i,4}}\frac{{{{\dot \mu }_2}(t)}}{{{\mu _2}(t)}}B_i^ + {e_i},
\end{align}
\begin{align}\label{ui2}
{u_{i,2}} =  - \frac{{{\rho _1}B_i^T{P_i}{e_i}}}{{\left\| {B_i^T{P_i}{e_i}} \right\|}} - \frac{{{\rho_2}B_i^T{P_i}{e_i}\left\| {{x_i}} \right\|}}{{\left\| {B_i^T{P_i}{e_i}} \right\|}} - \frac{{{\rho _3}B_i^T{P_i}{e_i}\left\| {{F_i}} \right\|}}{{\left\| {B_i^T{P_i}{e_i}} \right\|}},
\end{align}

\begin{align}\label{ui3}
{u_{i,3}} = & - {\rho_4}\Big(\frac{{B_i^T{P_i}{e_i}\left\| {{K_{1i}}{x_i}} \right\|}}{{\left\| {B_i^T{P_i}{e_i}} \right\|}} + \frac{{B_i^T{P_i}{e_i}\left\| {{K_{2i}}{\xi _i}} \right\|}}{{\left\| {B_i^T{P_i}{e_i}} \right\|}} \notag \\
& + \frac{{B_i^T{P_i}{e_i}\left\| {{K_{3i}}{{\tilde h}_i}} \right\|}}{{\left\| {B_i^T{P_i}{e_i}} \right\|}}+\frac{{{K_{4i}}B_i^T{P_i}{e_i}\left\| {\frac{{{{\dot \mu }_2}(t)}}{{{\mu _2}(t)}}B_i^ + {e_i}} \right\|}}{{\left\| {B_i^T{P_i}{e_i}} \right\|}}\Big),
\end{align}
where
\begin{align}\label{ei}
{e_i} = {x_1} - {X_i}{\xi _i} - {X_{hi}}{\tilde h_i},
\end{align}
and ${\rho _i}(i = 1,...,4)$  are the positive constants satisfying
\begin{align}\label{parameter_controller_selection}
\left\{ {\begin{array}{*{20}{l}}
{{\rho _1} \ge \frac{{(1 + \bar M)\bar d}}{{1 - \bar M}}},\\
{{\rho _2} \ge \frac{{\bar N}}{{1 - \bar M}}},\\
{{\rho _3} \ge \frac{{\bar u}}{{1 - \bar M}}},\\
{{\rho _4} \ge \frac{{\bar M}}{{1 - \bar M}}},
\end{array}} \right.
\end{align}
Besides, the gain matrix $K_{i,1}$ is chosen as ${K_{i,1}} =  - B_i^T{P_i}$ and the positive-definite matrix ${P_i}$ is the solution of the following Riccati equation
\begin{align}\label{likati2}
A_i^T{P_i} + {P_i}{A_i} - {P_i}{B_i}B_i^T{P_i} + {I_{{n_i}}} = 0.
\end{align}
The gain matrices $K_{i,2}$  and $K_{i,3}$  are chosen as ${K_{i,2}} = {U_i} - {K_{i,1}}{X_i}$  and ${K_{i,3}} = {U_{hi}} - {K_{i,1}}{X_{hi}}$.  ${K_{i,4}}$ is a positive constant. $\mu_2(t) (t\geq T_1)$ is a time scaling function defined as
\begin{align}\label{mu2}
{\mu _2}(t) = \left\{ {\begin{array}{*{20}{l}}
{\frac{{T_2^v}}{{{{\left( {{T_2} + {T_1} - t} \right)}^v}}},}&{t \in \left[ {{T_1},{T_1} + {T_2}} \right)}\\
{1,}&{t \in \left[ {{T_1} + {T_2},\infty } \right)}
\end{array}} \right.
\end{align}
where $T_2$ is the prescribed convergence time of the whole controlled system satisfying ${T_2} > {T_1}$.\par
Then, the following theorem is given to describe the stability of the controlled system under the proposed formation control strategy.
\begin{theorem}\label{thm2}
 Considering the multi-agent system (\ref{model_follower}) and (\ref{leader_model}), suppose the Assumptions 1-5 hold, then under the observer (\ref{observer}) and the controller (\ref{u}) with parameters satisfying (\ref{parameter_controller}) and (\ref{parameter_controller_selection}),
the output formation errors $\bar e_i = y_i-y_0-h_i (i=1,...,N)$ converge to the origin in the prescribed time $T_2$.
\end{theorem}
\begin{proof}
For $t \le {T_1}$, ${u_i} = \textbf{0}$, it is easy to find that all states in the controlled system are bounded.
Then, for $t > {T_1}$, it holds from Theorem 1 that ${\xi _i} = {x_0}$. Meanwhile, from  (\ref{model_follower}) and  (\ref{ui1}), the derivative of ${e_i}$  satisfies
\begin{align}\label{de}
{{\dot e}_i} = &({A_i} + \Delta {A_i}){x_i} + ({B_i} + \Delta {B_i})({u_i} + {d_i})\notag \\
& - {X_i}({A_0}{x_0} + {B_0}{u_0}) - {X_{hi}}{A_{hi}}{{\tilde h}_i}\notag \\
= & {A_i}{x_i} + B_i{u_{i,1}} - {X_i}({A_0}{x_0} + {B_0}{u_0}) - {X_{hi}}{A_{hi}}{{\tilde h}_i}\notag \\ &+ ({B_i} + \Delta {B_i}){d_i}+ \Delta {A_i}{x_i} + \Delta {B_i}{u_{i,1}} + ({B_i} +\notag \\
& \Delta {B_i})({u_{i,2}} + {u_{i,3}})\notag \\
{\rm{    =  (}}&{A_i} + B_i{K_{i,1}}){e_i} + (A _i+ B_i{K_{i,1}})X_i{x_0} + B_i{K_{i,2}}{x_0} \notag \\
&- {X_i}{A_0}{x_0} + (A_i + B_i{K_{i,1}})X_{hi}{{\tilde h}_i} + B_i{K_{i,3}}{{\tilde h}_i} \notag \\
& - {X_{hi}}{A_{hi}}{{\tilde h}_i}- {X_i}{B_0}{u _0} + ({B_i} + \Delta {B_i}){d_i} + \Delta {A_i}{x_i}  \notag \\
& + \Delta {B_i}{u_{i,1}} + ({B_i} + \Delta {B_i})({u_{i,2}} + {u_{i,3}})- {K_{i,4}}\frac{{{{\dot \mu }_2}(t)}}{{{\mu _2}(t)}}{e_i}\notag \\
{\rm{  }} = & {\rm{(}}{A_i} + B_i{K_{i,1}}){e_i} - {X_i}{B_0}{u _0} + ({B_i} + \Delta {B_i}){d_i} + \Delta {A_i}{x_i}\notag \\
{\rm{      }} &  + ({B_i} + \Delta {B_i})({u_{i,2}} + {u_{i,3}})+ \Delta {B_i}{u_{i,1}} - {K_{i,4}}\frac{{{{\dot \mu }_2}(t)}}{{{\mu _2}(t)}}{e_i}.
\end{align}\par
For $t>T_1$, choose the Lyapunov candidate for follower $i\;(i=1,...,N)$ as
\begin{align}\label{barV}
{\bar V_i} = \frac{1}{2}e_i^T{P_i}{e_i}.
\end{align}\par
Then, from (\ref{ui1}), the derivative of ${\bar V_i}$ along the trajectory (\ref{de}) satisfies
\begin{align}\label{dV}
{{\dot {\bar V}}_i} =& e_i^T{P_i}{{\dot e}_i}\notag \\
 =& {e_i}{P_i}(A_i+ B_i{K_{i,1}}){e_i} - {K_{i,4}}\frac{{{{\dot \mu }_2}(t)}}{{{\mu _2}(t)}}e_i^T{P_i}{e_i}  \notag \\
 &+ e_i^T{P_i}\Delta {A_i}{x_i}+ e_i^T{P_i}({B_i} + \Delta {B_i}){d_i}\notag \\
&  - e_i^T{P_i}{X_i}{B_0}{u_0}+ e_i^T{P_i}({B_i} + \Delta {B_i})({u_{i,2}} + {u_{i,3}})  \notag \\
& + e_i^ T {P_i}\Delta {B_i}{u_{i,1}}\notag \\
 \le & \frac{1}{2}{e_i}({P_i}(A _i+ B_i{K_{i,1}}){ + }{(A_i + B_i{K_{i,1}})^T}{P_i}){e_i}\notag \\
  &- {K_{i,4}}\frac{{{{\dot \mu }_2}(t)}}{{{\mu _2}(t)}}e_i^T{P_i}{e_i} + \left\| {e_i^T{P_i}{B_i}} \right\|\left\| {{N_i}(t)} \right\|\left\| {{x_i}} \right\| \notag \\
& \left\| {e_i^T{P_i}{B_i}} \right\|(1 + \left\| {M(t)} \right\|)\bar d - e_i^T{P_i}{B_i}{F_i}{u_0} \notag \\&+ e_i^T{P_i}({B_i} + \Delta {B_i})({u_{i,2}} + {u_{i,3}})
+ e_i^ T {P_i}\Delta {B_i}{u_{i,1}}.
\end{align}\par
Substituting (\ref{ui2}) and (\ref{likati2}) into (\ref{dV}), it gives
\begin{align}\label{dV2}
{{\dot {\bar V}}_i} \le  & - \frac{1}{2}e_i^T{e_i} - {K_{i,4}}\frac{{{{\dot \mu }_2}(t)}}{{{\mu _2}(t)}}e_i^T{P_i}{e_i} + \bar N\left\| {e_i^T{P_i}{B_i}} \right\|\left\| {{x_i}} \right\|\notag \\
& \left\| {e_i^T{P_i}{B_i}} \right\|(1 + \bar M)\bar d + \left\| {e_i^T{P_i}{B_i}} \right\|\left\| {{F_i}} \right\|\left\| {{u_0}} \right\| \notag \\
& + e_i^T{P_i}{B_i}\Big( - \frac{{{\rho _1}{B_i^T}{P_i}{e_i}}}{{\left\| {{B_i^T}{P_i}{e_i}} \right\|}} - \frac{{{\rho _2}\left\| {{x_i}} \right\|{B_i^T}{P_i}{e_i}}}{{\left\| {{B_i^T}{P_i}{e_i}} \right\|}}\notag \\
& - \frac{{{\rho _3}\left\| {{F_i}} \right\|{B_i^T}{P_i}{e_i}}}{{\left\| {{B_i^T}{P_i}{e_i}} \right\|}}\Big)
+ e_i^T{P_i}\Delta {B_i}\Big( - \frac{{{\rho_1}{B_i}{P_i}{e_i}}}{{\left\| {{B_i^T}{P_i}{e_i}} \right\|}} \notag \\
&- \frac{{{\rho _2}\left\| {{x_i}} \right\|{B_i^T}{P_i}{e_i}}}{{\left\| {{B_i^T}{P_i}{e_i}} \right\|}} - \frac{{{\rho _3}\left\| {{F_i}} \right\|{B_i^T}{P_i}{e_i}}}{{\left\| {{B_i^T}{P_i}{e_i}} \right\|}}\Big) \notag \\& + e_i^T{P_i}({B_i} + \Delta {B_i}){u_{i,3}}+ e_i^ T {P_i}\Delta {B_i}{u_{i,1}}.
\end{align}\par
It can be obtained from (\ref{parameter_controller_selection}) that
\begin{align}\label{fangsuo1}
& \left\| {e_i^T{P_i}{B_i}} \right\|(1 + \bar M)\bar d - e_i^T{P_i}{B_i}\frac{{{\rho _1}{B_i^T}{P_i}{e_i}}}{{\left\| {{B_i}{P_i}{e_i}} \right\|}}\notag \\
&  - e_i^T{P_i}\Delta {B_i}\frac{{{\rho _1}{B_i^T}{P_i}{e_i}}}{{\left\| {{B_i}{P_i}{e_i}} \right\|}}\notag \\
 \le  & - {\rho _1}\left\| {{B_i^T}{P_i}{e_i}} \right\| + \left\| {{B_i^T}{P_i}{e_i}} \right\|{\rho _1}\bar M \notag \\
 &+ \left\| {e_i^T{P_i}{B_i}} \right\|(1 + \bar M)\bar d \notag\\
 \le &  - \left\| {{B_i^T}{P_i}{e_i}} \right\|({\rho_1}(1 - \bar M) - (1 + \bar M)\bar d) \le  0,
\end{align}
\begin{align}\label{fangsuo2}
&\bar N\left\| {e_i^T{P_i}{B_i}} \right\|\left\| {{x_i}} \right\| - e_i^T{P_i}{B_i}\frac{{{\rho _2}\left\| {{x_i}} \right\|{B_i^T}{P_i}{e_i}}}{{\left\| {{B_i^T}{P_i}{e_i}} \right\|}} \notag \\
&- e_i^T{P_i}\Delta {B_i}\frac{{{\rho _2}\left\| {{x_i}} \right\|{B_i^T}{P_i}{e_i}}}{{\left\| {{B_i}{P_i}{e_i}} \right\|}} \notag\\
 \le& \bar N\left\| {e_i^T{P_i}{B_i}} \right\|\left\| {{x_i}} \right\| - {\rho _2}\left\| {{x_i}} \right\|\left\| {e_i^T{P_i}{B_i}} \right\| \notag \\
 & + {\rho_2}\bar M\left\| {{x_i}} \right\|\left\| {e_i^T{P_i}{B_i}} \right\| \notag \\
 \le & - \left\| {{x_i}} \right\|\left\| {e_i^T{P_i}{B_i}} \right\|({\rho _2}(1 - \bar M) - \bar N)
 \le 0,
\end{align}
and
\begin{align}\label{fangsuo3}
& \left\| {e_i^T{P_i}{B_i}} \right\|\left\| {{F_i}} \right\|\left\| {{u_0}} \right\| - e_i^T{P_i}{B_i}\frac{{{\rho  _3}\left\| {{F_i}} \right\|{B_i^T}{P_i}{e_i}}}{{\left\| {{B_i^T}{P_i}{e_i}} \right\|}} \notag \\
&- e_i^T{P_i}\Delta {B_i}\frac{{{\rho _3}\left\| {{F_i}} \right\|{B_i^T}{P_i}{e_i}}}{{\left\| {{B_i}{P_i}{e_i}} \right\|}} \notag \\
 \le & \left\| {e_i^T{P_i}{B_i}} \right\|\left\| {{F_i}} \right\|\left\| {{u_0}} \right\| - {\rho _3}\left\| {e_i^T{P_i}{B_i}} \right\|\left\| {{F_i}} \right\| \notag \\
 &+ {\rho_3}\bar M\left\| {e_i^T{P_i}{B_i}} \right\|\left\| {{F_i}} \right\|\notag \\
 \le & - \left\| {e_i^T{P_i}{B_i}} \right\|\left\| {{F_i}} \right\|({\rho _3}(1 - \bar M) - \left\| {{u_0}} \right\|)\notag \\
 \le &  - \left\| {e_i^T{P_i}{B_i}} \right\|\left\| {{F_i}} \right\|({\rho _3}(1 - \bar M) - \bar u) \le 0.
\end{align}
Then, substituting (\ref{fangsuo1})-(\ref{fangsuo3}) into (\ref{dV2}) leads to
\begin{align}\label{dv3}
{\dot {\bar {V}}_i} \le &  - \frac{1}{2}e_i^T{e_i} - {K_{i,4}}\frac{{{{\dot \mu }_2}(t)}}{{{\mu _2}(t)}}e_i^T{P_i}{e_i} + e_i^T{P_i}({B_i} + \Delta {B_i}){u_{i,3}} \notag \\
& + e_i^ T {P_i}\Delta {B_i}{u_{i,1}}.
\end{align}
Further, by combining (\ref{ui3}) and (\ref{dv3}), one can get
\begin{align}\label{buqueding}
\dot {\bar V}_i \le&    - \frac{1}{2}e_i^T{e_i} - {K_{i,4}}\frac{{{{\dot \mu }_2}}}{{{\mu _2}}}e_i^T{P_i}{e_i} + e_i^T{P_i}{(B_i+ \Delta B_i)}\Big(  - {\rho _4}\cdot \notag \\
&\Big(\frac{{B_i^T{P_i}{e_i}\left\| {{K_{i,1}}{x_i}} \right\|}}{{\left\| {B_i^T{P_i}{e_i}} \right\|}}+ \frac{{B_i^T{P_i}{e_i}\left\| {{K_{i,2}}{\xi _i}} \right\|}}{{\left\| {B_i^T{P_i}{e_i}} \right\|}}\notag \\
&+ \frac{{{K_{i,4}}B_i^T{P_i}{e_i}\left\| {\frac{{{{\dot \mu }_2}(t)}}{{{\mu _2}(t)}}B_i^ + {e_i}} \right\|}}{{\left\| {B_i^T{P_i}{e_i}} \right\|}}\Big)\Big) + e_i^ T {P_i}\Delta {B_i}{u_{i,1}}\notag  \\
 \le & - \frac{1}{2}e_i^T{e_i} - {K_{i,4}}\frac{{{{\dot \mu }_2}(t)}}{{{\mu _2}(t)}}e_i^T{P_i}{e_i} \notag \\
 & - {\rho _4}\left\| {B_i^T{P_i}{e_i}} \right\|\left\| {{K_{i,1}}{x_i}} \right\| \notag \\
 & - {\rho_4}\left\| {B_i^T{P_i}{e_i}} \right\|\left\| {{K_{i,2}}{\xi _i}} \right\|  - {\rho _4}\left\| {B_i^T{P_i}{e_i}} \right\|\left\| {{K_{i,3}}{{\tilde h}_i}} \right\|\notag \\
 &- {\rho_4}{K_{i,4}}\left\| {B_i^T{P_i}{e_i}} \right\|\left\| {\frac{{{{\dot \mu }_2}(t)}}{{{\mu _2}(t)}}B_i^ + {e_i}} \right\| \notag \\
 &+ {\rho_4}\bar M\left\| {B_i^T{P_i}{e_i}} \right\|\left\| {{K_{i,1}}{x_i}} \right\|\notag \\ &+{\rho_4}\bar M\left\| {B_i^T{P_i}{e_i}} \right\|\left\| {{K_{i,2}}{\xi _i}} \right\| \notag \\
 & + {\rho_4}\bar M\left\| {B_i^T{P_i}{e_i}} \right\|\left\| {{K_{i,3}}{{\tilde h}_i}} \right\| \notag \\
 & +{\rm{  }}{\rho_4}\bar M\left\| {B_i^T{P_i}{e_i}} \right\|\left\| {\frac{{{{\dot \mu }_2}(t)}}{{{\mu _2}(t)}}B_i^ + {e_i}} \right\| \notag \\
 & + \bar M\left\| {B_i^T{P_i}{e_i}} \right\|\left\| {{K_{i,1}}{x_i}} \right\|
 + \bar M\left\| {B_i^T{P_i}{e_i}} \right\|\left\| {{K_{i,2}}{\xi _i}} \right\| \notag \\ &+ \bar M\left\| {B_i^T{P_i}{e_i}} \right\|\left\| {{K_{i,3}}{{\tilde h}_i}} \right\| \notag \\
& {\rm{     + }}\bar M{K_{i,4}}\left\| {B_i^T{P_i}{e_i}} \right\|\left\| {\frac{{{{\dot \mu }_2}(t)}}{{{\mu _2}(t)}}B_i^ + {e_i}} \right\|\notag \\
= & -\frac{1}{2}e_i^T{e_i} - {K_{i,4}}\frac{{{{\dot \mu_2 }{(t)}}}}{{{\mu _2(t)}}}e_i^T{P_i}{e_i} - \Big(\left\| {B_i^T{P_i}{e_i}} \right\|\left\| {{K_{i,1}}{x_i}} \right\| \notag \\
& + \left\| {B_i^T{P_i}{e_i}} \right\|\left\| {{K_{i,2}}{\xi _i}} \right\|
 + \left\| {B_i^T{P_i}{e_i}} \right\|\left\| {{K_{i,3}}{{\tilde h}_i}} \right\| +\notag  \\
& {K_{i,4}}\left\| {B_i^T{P_i}{e_i}} \right\|\left\| {\frac{{{{\dot \mu }_2}(t)}}{{{\mu _2}(t)}}B_i^ + {e_i}} \right\|\Big)({\rho_4}(1 - \bar M) - \bar M).
\end{align}\par
Therefore, according to (\ref{parameter_controller_selection}) and (\ref{buqueding}), it holds that
\begin{align}
\dot{ \bar V}_i \le &  - \frac{1}{2}e_i^T{e_i} - {K_{i,4}}\frac{{{{\dot \mu }_2}(t)}}{{{\mu _2}(t)}}e_i^T{P_i}{e_i}\notag \\
 \le &  - \frac{{{\bar V_i}}}{{{\lambda _{\max }}({P_i})}} - {2K_{i,4}}\frac{{{{\dot \mu }_2}(t)}}{{{\mu _2}(t)}}{\bar V_i}.
\end{align}\par
Per Lemma 2, it is obtained that ${e_i}$ converges to the origin in the prescribed time $T_2$. Then, From (\ref{model_follower}), (\ref{leader_model}) and (\ref{ei}), one has that
\begin{align}\label{output_error}
 \bar e_i =& {y_i} - {y_0} - {h_i} \notag \\
 = & {C_i}{x_i} - {C_0}{x_0} - {C_{hi}}{{\tilde h}_i} \notag \\
 = &{C_i}({x_i} - {X_i}{x_0} - {X_{hi}}{{\tilde h}_i}) \notag \\
 = &  {C_i}{e_i}.
\end{align}\par
It has been proved that for $t > {T_2}$, ${e_i} = 0$, thus $\bar e_i =0$ holds for $t > {T_2} $. That is, the formation error  $\bar e_i$ converges to the origin in the prescribed time  $T_2$.\par
The proof is thus completed.
\end{proof}

\begin{remark}
When $\left\| {B_i^T{P_i}{e_i}} \right\|=0$, the $u_{i,2}$ and $u_{i,3}$ can be replaced by $\textbf{0}$ and it's easy to prove that the Theorem 2 still holds. The detailed proof is omitted here due to the limited space.
\end{remark}

\section{Simulation}
Considering a group of MAS composed of five followers and one leader, the interaction network among agents is depicted in Fig.\;\ref{The interaction network among agents}.
\begin{figure}
  \centering
  \includegraphics[width= 80pt]{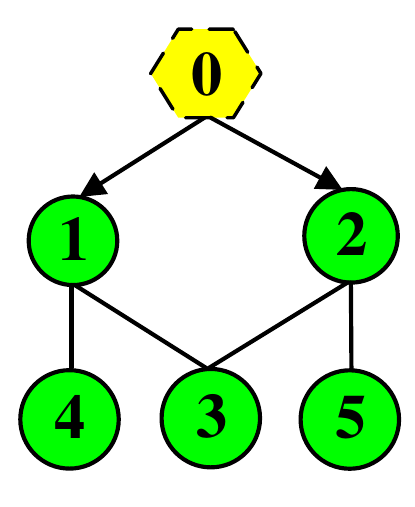}
  \caption{The interaction network among agents.}\label{The interaction network among agents}
\end{figure}
Consider the following dynamics of followers: ${A_1} = \left[ {\begin{array}{*{20}{l}}
0&1&0\\
0&2&1\\
6&0&0
\end{array}} \right]$, ${B_1} = \left[ {\begin{array}{*{20}{l}}
{0.1}&0&1\\
1&0&1\\
{ - 1}&1&0
\end{array}} \right]$,
${A_2} = \left[ {\begin{array}{*{20}{l}}
0&2&0\\
0&3&1\\
6&0&0
\end{array}} \right]$, ${B_2} = \left[ {\begin{array}{*{20}{l}}
{0.1}&0&1\\
2&0&0\\
{ - 1}&1&1
\end{array}} \right]$,
${A_3} = \left[ {\begin{array}{*{20}{l}}
0&3&0\\
0&0&1\\
6&0&0
\end{array}} \right]$, ${B_3} = \left[ {\begin{array}{*{20}{l}}
{-2}&0&-1\\
1&0.2&0\\
{ - 1}&0&1
\end{array}} \right]$,
${C_1} =C_2= C_3 = \left[ {\begin{array}{*{20}{l}}
{1}&0&0\\
1&2&0\\
\end{array}} \right]$,
${A_4} = \left[ {\begin{array}{*{20}{l}}
0&1\\
0&0\\
\end{array}} \right]$, ${B_4} = \left[ {\begin{array}{*{20}{l}}
{0.1}&0\\
-1&1
\end{array}} \right]$,
${A_5} = \left[ {\begin{array}{*{20}{l}}
0&2\\
0&0\\
\end{array}} \right]$, ${B_5} = \left[ {\begin{array}{*{20}{l}}
{1.1}&0\\
-1&1\\
\end{array}} \right]$.
${C_4} =C_5 = \left[ {\begin{array}{*{20}{l}}
{1}&0\\
0&1\\
\end{array}} \right]$,
The leader is modeled by
 ${A_0} = \left[ {\begin{array}{*{20}{l}}
0&0&1\\
-0.6&0&1\\
0&0&0
\end{array}} \right]$, ${B_0} = \left[ {\begin{array}{*{20}{l}}
{0}\\
0\\
1
\end{array}} \right]$,  and ${C_0} = \left[ {\begin{array}{*{20}{l}}
1&0&0\\
0&1&0\\
\end{array}} \right]$.\par
The uncertainties are set as\\
${\Delta A_1} = \left[ {\begin{array}{*{20}{l}}
0.005&0&0.1-0.005\sin(t/3)\\
0.05&0&0.1-0.05\sin(t/3)\\
-0.05&0.01&0.05\sin(t/3)-0.05
\end{array}} \right]$,\;
${\Delta B_1} = \left[ {\begin{array}{*{20}{l}}
{0.002}&0&0.01\\
0.02&0&0.01\\
{ - 0.02}&-0.01\sin(t/3)&0
\end{array}} \right]$,\\
${\Delta A_2} = \left[ {\begin{array}{*{20}{l}}
-0.005&1&0.001\\
-0.1&0&0.02\\
0.05&0.05&-0.09
\end{array}} \right]$, \;
${\Delta B_2} = \left[ {\begin{array}{*{20}{l}}
{-0.012}&0&0.01\\
-0.04&0&0\\
{ 0.02}&-0.01&0.01
\end{array}} \right]$,\\
${\Delta A_3} = \left[ {\begin{array}{*{20}{l}}
0&-0.1&-0.2\\
0&-0.01&0.09\\
0&0.1&-0.1
\end{array}} \right]$, \;
${\Delta B_3} = \left[ {\begin{array}{*{20}{l}}
{0.03}&0&-0.01\\
-0.008&-0.002\cos(t/2)&0\\
{ 0}&0&0.01
\end{array}} \right]$,\\
${\Delta A_4} = \left[ {\begin{array}{*{20}{l}}
-0.01&0.005\cos(t/2)\\
0.16&-0.05\cos(t/2)+0.1\\
\end{array}} \right]$, \\
${\Delta B_4} = \left[ {\begin{array}{*{20}{l}}
{-0.001}&0\\
-0.005&0\\
\end{array}} \right]$,\;
${\Delta A_5} = \left[ {\begin{array}{*{20}{l}}
0.11&-0.055\sin(t)\\
-0.15&-0.05\sin(t)+0.1\\
\end{array}} \right]$, \\
${\Delta B_5} = \left[ {\begin{array}{*{20}{l}}
{0.011\cos(t)}&0\\
-0.01\cos(t)+0.011&0\\
\end{array}} \right]$.\par
To generate the desired formation pattern, the exosystem (\ref{formation_system}) is set as\\
${ A_{hi}} = \left[ {\begin{array}{*{20}{l}}
0&1&0&0\\
-1&0&0&0\\
0&0&0&1\\
0&0&-1&0\\
\end{array}} \right]$,\;
${ C_{hi}} = \left[ {\begin{array}{*{20}{l}}
1&0&0&0\\
0&0&1&0
\end{array}} \right]$,   $i=1,...,5.$\\
Besides,
\begin{align}
 \tilde h_i(0) = 4[\cos((2(i-1)\pi/5),  -2\sin((2(i-1)\pi/5),  \notag \\ 2\sin((2(i-1)\pi/5), 2\cos((2(i-1)\pi/5)]]^T.
 \end{align}\par
 The external disturbances are set as  $d_1 =[\sin(t),\cos(t),\sin(t/2)]^T$, $d_2=[2\sin(t),\cos(t),0.5]^T$, $d_3=[\sin(t),2\cos(t),0]^T$, $d_4=[2\text{{exp}}(-3t), \sin(t)]^T$ and $d_5=[2\text{exp}(-3t), \cos(2t)]^T$.
 The input of the leader is $u_0 =\sin(t/2)$.\par
The parameters in the observers and the controllers are chosen as $\alpha_i = 5$, $\beta_i = 0.2$, $c_i = 10\, (i=1,...,5)$, $\rho_1 = 5 $, $\rho_2 = 5$, $\rho_3 = 0.2$, $\rho_4 = 0.03$, $T_1= 0.5$, $T_2= 4$. \par
Fig.\;\ref{Output snapshots of agents at different times} shows the outputs $y_i=[y_{i1}, y_{i2}] \in R^2$ of all agents. The formation errors $\bar e_i=[\bar e_{i1}, \bar e_{i2}] \in R^2 \;(i=1,...,N)$ are given in Fig.\;\ref{The formation errors}. It's obvious that all formation errors converge to the origin in the prescribed time $T_2$, which reflects the effectiveness of the proposed control approach.
\begin{figure}[]
 \centering
 \subfloat[Output snapshot at 1s.]
  {
  \includegraphics[width=120pt, height=100pt]{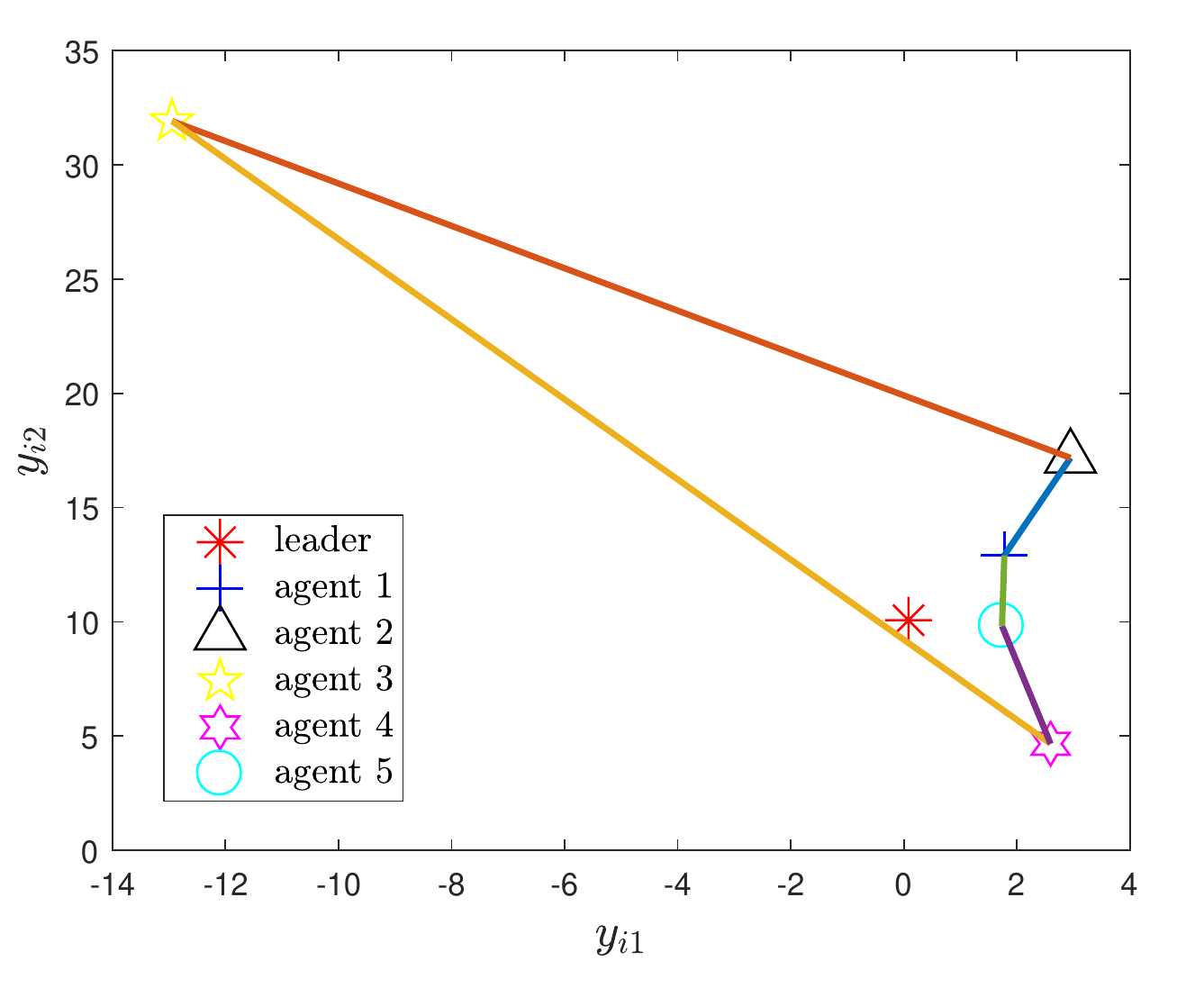}
  }
 \subfloat[Output snapshot at 2s.]
  {
  \includegraphics[width=120pt, height=100pt]{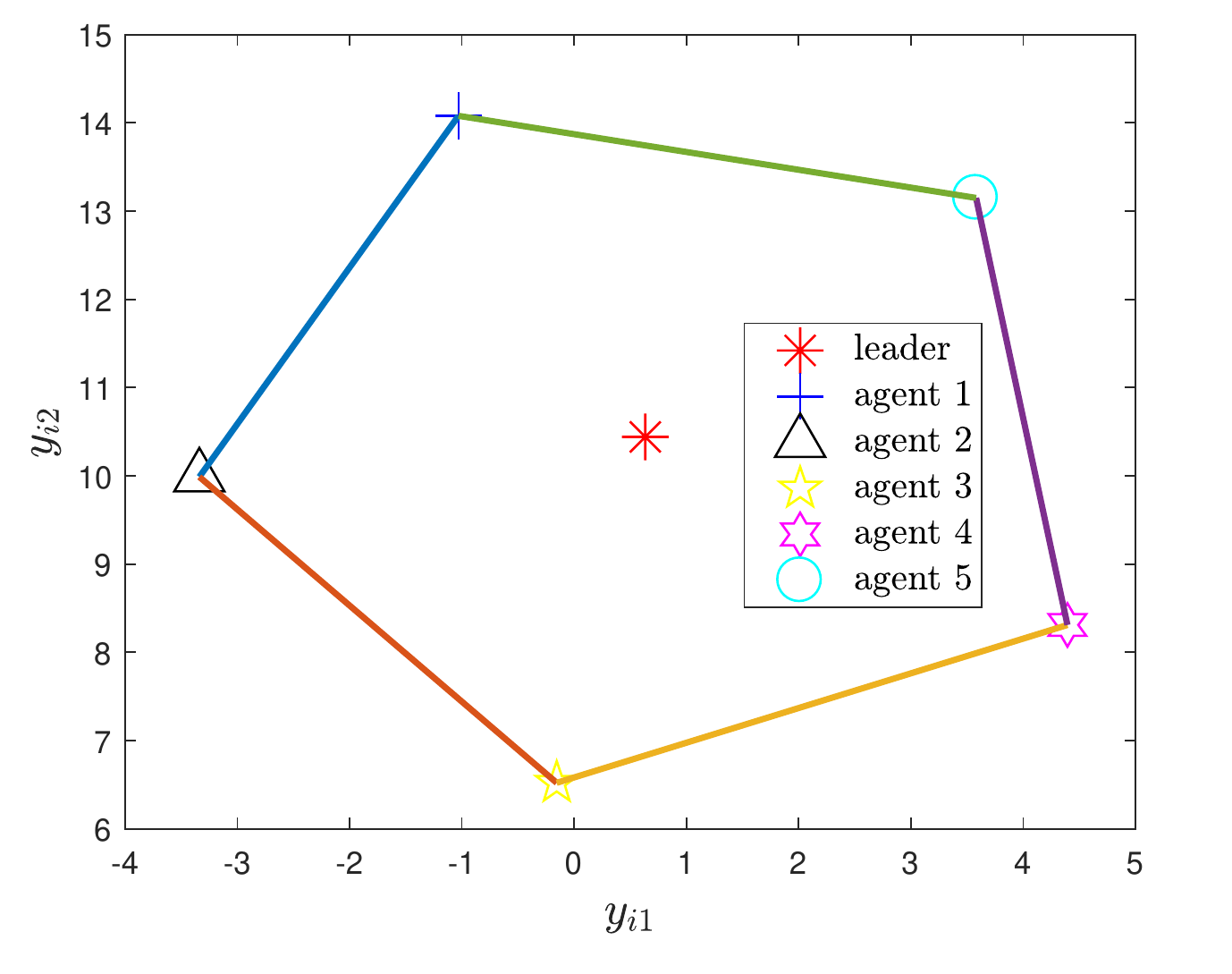}
  }\\
  \subfloat[Output snapshot at 4s.]
  {
  \includegraphics[width=120pt, height=100pt]{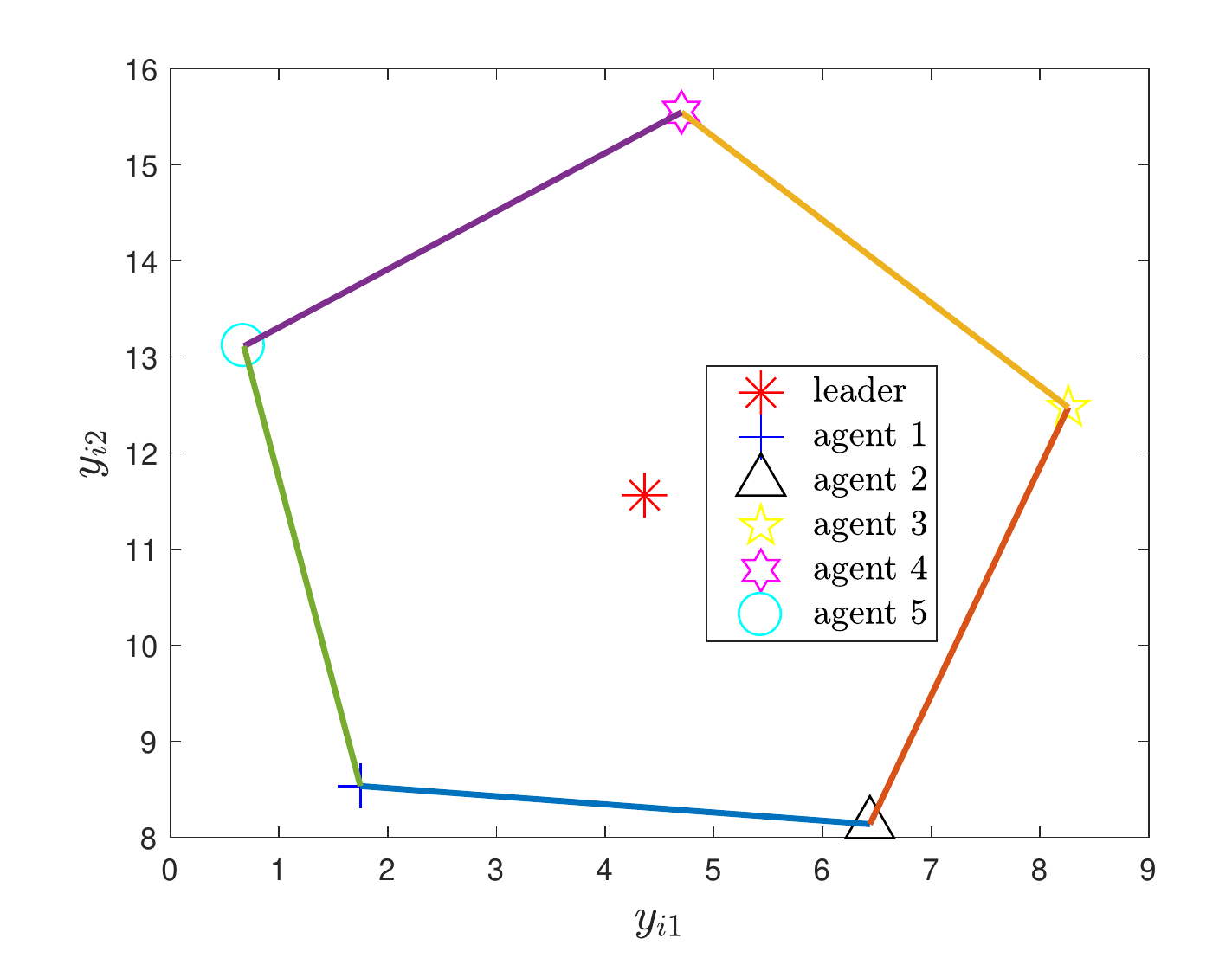}
  }
 \subfloat[Output snapshot at 10s.]
  {
  \includegraphics[width=120pt, height=100pt]{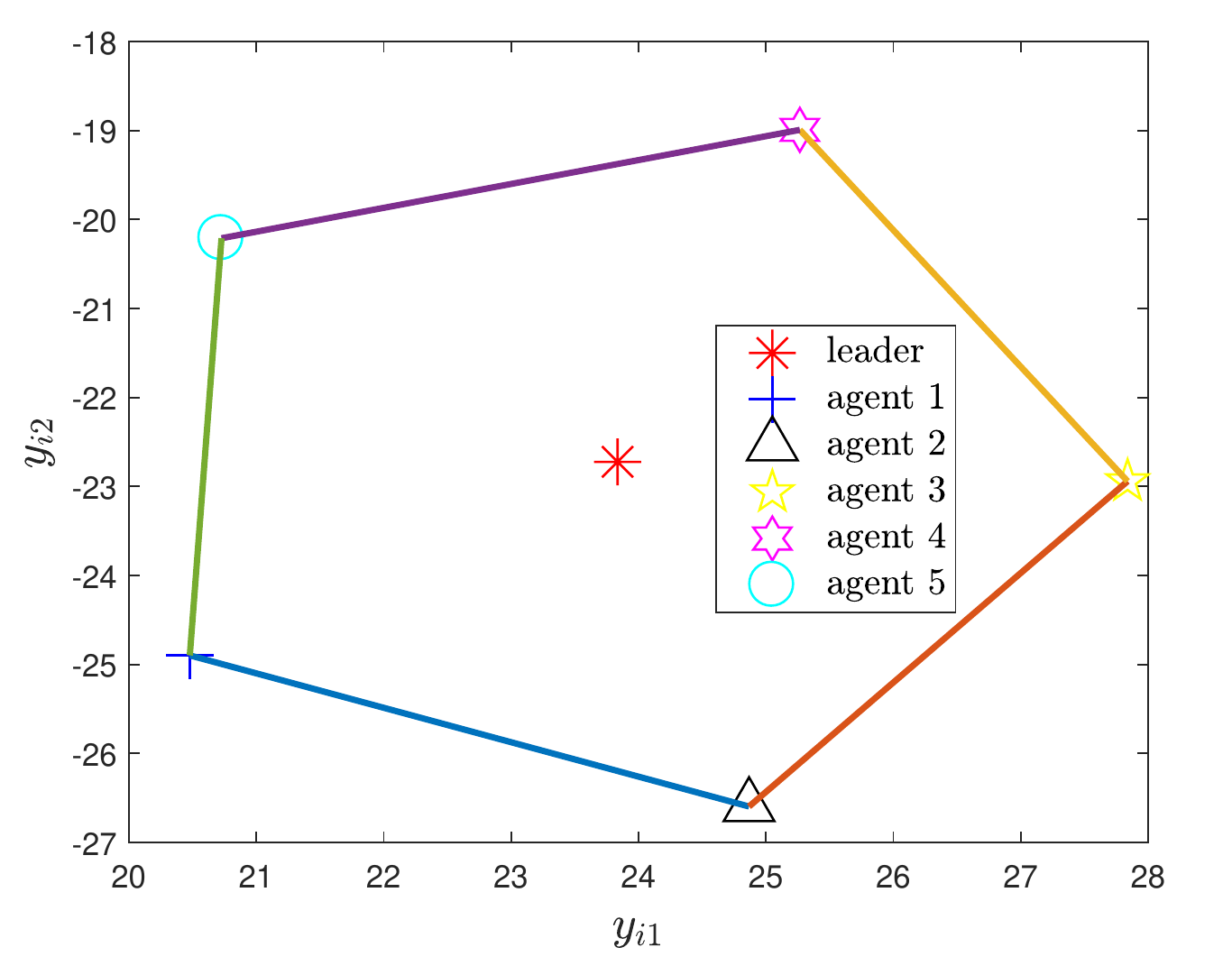}
  }
  \caption{Output snapshots of agents at different times.}\label{Output snapshots of agents at different times}
\end{figure}
\begin{figure}
  \centering
  \includegraphics[width=200pt]{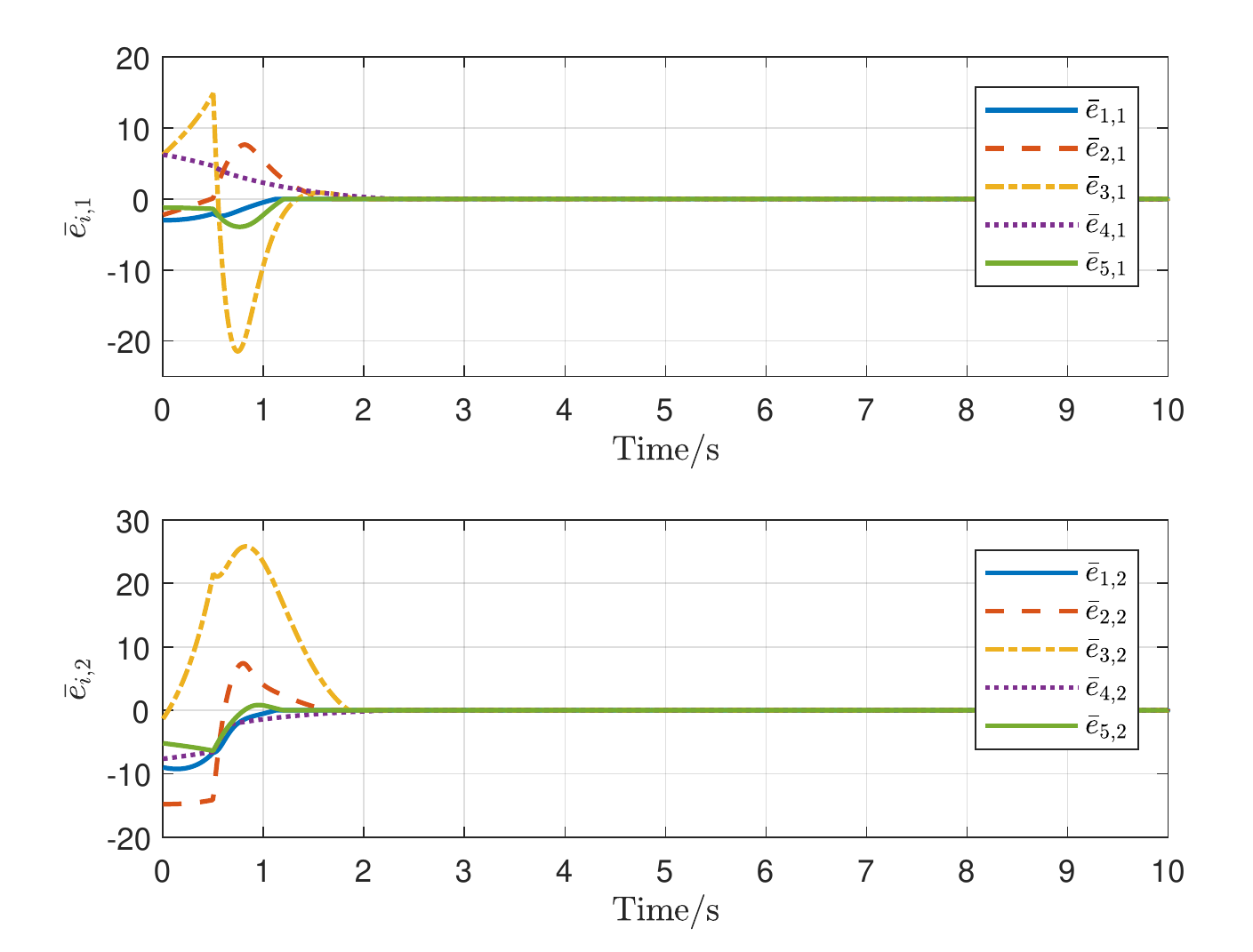}
  \caption{The formation errors $\bar e_{i,m}\;(i=1,...,5; m=1,2).$ }\label{The formation errors}
\end{figure}

\section{Conclusion}
The prescribed time time-varying output formation tracking of the heterogeneous MAS has been investigated in this paper. First, a distributed observer is proposed to estimate the states of the leader, which is able to provide an accurate estimation in the prescribed time. Then, the prescribed time formation controller is designed for each follower considering the disturbances and the uncertainties in both the state matrix and input matrix. Further, the effectiveness of the proposed control strategy is demonstrated by the theory analysis and simulation results.
\section*{Acknowledgements}
This work is supported by National Science Foundation of China under Grant 61873319.

\bibliographystyle{elsarticle-num}
\bibliography{Prescribed_Time_Time_varying_Output_Formation_Tracking_for_Uncertain_Heterogeneous_Multi_agent_systems}


\end{document}